%% file: root.tex
\documentclass[letterpaper, 10 pt, conference]{ieeeconf}  % Comment this line out if you need a4paper

\IEEEoverridecommandlockouts                              % This command is only needed if 
\usepackage[cmex10]{amsmath}
\usepackage{amssymb}  % assumes amsmath package installed
\usepackage{subfigure}

\usepackage{pdfpages}

\usepackage{ stmaryrd }

\usepackage{cite}

\usepackage[normalem]{ulem}

\newcommand{\df}{\mathrel{\mathop:}=}

\newcommand{\RNum}[1]{\uppercase\expandafter{\romannumeral #1\relax}}

\newtheorem{theo}{Theorem}
\newtheorem{rem}{Remark}
\newtheorem{lem}{Lemma}

\newcommand{\ie}{i\/.\/e\/.,\/~}
\newcommand{\eg}{e\/.\/g\/.,\/~}
\newcommand{\cf}{cf\/.\/~}
\newcommand{\fig}{Fig\/.\/~}
\newcommand{\sect}{Sec\/.\/~}

\usepackage{fancyhdr}
\newcommand{\mytitle}{\textbf{Accepted final version.}
To appear in \textit{Proc. of the 2018 American Control Conference (ACC)}.}
\title{\LARGE \bf
Event-triggered Learning for Resource-efficient Networked Control
}

\author{Friedrich Solowjow$^{1,2}$, Dominik Baumann$^{1}$, Jochen Garcke$^{2,3}$, Sebastian
Trimpe$^{1}$
\thanks{$^{1}$Intelligent Control Systems Group, Max Planck Institute for Intelligent Systems, 70569 Stuttgart, Germany. Email: \{fsolowjow, dbaumann, strimpe\}@tuebingen.mpg.de}%
\thanks{$^{2}$Inst. for Numerical Simulation, University of Bonn, 53115 Bonn, Germany.}%
\thanks{$^{3}$Fraunhofer SCAI, Schloss Birlinghoven, 53754 St. Augustin, Germany.}%
\thanks{This work was supported in part by the German Research Foundation (DFG)
Priority Program 1914 (grant TR 1433/1-1), the Cyber Valley Initiative, and the Max Planck Society.}% <-this % stops a space
}

\begin{document}
	\maketitle
	\thispagestyle{fancy}
	\pagestyle{empty}

%%%%%%%%%%%%%%%%%%%%%%%%%%%%%%%%%%%%%%%%%%%%%%%%%%%%%%%%%%%%%%%%%%%%%%%%%%%%%%%%
%%%%%%%%%%%%%%%%%%%%%%%%%%%%%%%%%%%%%%%%%%%%%%%%%%%%%%%%%%%%%%%%%%%%%%%%%%%%%%%%
	\input{texs/0abstract.tex}
	\input{texs/1introduction.tex}
	\input{texs/2preliminaries.tex}
	\input{texs/3learning.tex}
	\input{texs/4ETL.tex}
	\input{texs/5trigger.tex}
	\input{texs/6simulation.tex}

	\input{texs/7experiment.tex}
	\input{texs/8rest.tex}
	\bibliographystyle{IEEEtran}
\bibliography{Database}	
%%%%%%%%%%%%%%%%%%%%%%%%%%%%%%%%%%%%%%%%%%%%%%%%%%%%%%%%%%%%%%%%%%%%%%%%%%%%%%%%

\end{document}

%% file: texs/0abstract.tex
\begin{abstract}
Common event-triggered state estimation (ETSE) algorithms save communication in networked control systems by predicting agents' behavior, 
%by means of dynamics models, 
and transmitting updates only when the predictions deviate significantly.  
%How much communication can be reduced 
The effectiveness in reducing communication thus heavily depends on the quality of the dynamics models used to predict the agents' states or measurements.
%The effectiveness in reducing communication thus heavily relies on precise dynamics models used to predict other agents' states or measurements.
%Event-triggered learning is proposed herein as a novel concept, where model learning is triggered  to trigger model learning when needed and thus further reduce communications:
Event-triggered learning is proposed herein as a novel concept to 
further reduce communication:
%trigger model learning when needed:
%where learning of new models is triggered when needed in order 
%to further lower communication requirements:
whenever poor communication performance is detected, an identification experiment is triggered and an improved prediction model learned from data. 
Effective learning triggers are obtained by comparing the actual communication rate with the one that is expected based on the current model.  By analyzing statistical properties of the inter-communication times and leveraging powerful convergence results, the proposed trigger is proven to limit learning experiments to the necessary instants.
Numerical and physical experiments demonstrate that event-triggered learning improves robustness toward changing environments and yields lower communication rates than common ETSE.

% Prev version:
%Event-triggered learning is proposed as a novel concept to issue model learning experiments on necessity. The concept is applied to reduce communication in networked control systems, where multiple agents observe a dynamic process and sporadically transmit their state to other agents according to common event-triggered state estimation (ETSE) protocols. The effectiveness of ETSE relies heavily on precise dynamics models, since future actions are anticipated through model based state predictions. An additional trigger for model learning is suggested to infer model based expected communication behavior and quantify the model accuracy with the aid of observed data. To obtain an effective learning trigger, we propose to regard intercommunication times as random variables and leverage their statistical properties. This way, we can built upon powerful convergence results and ensure theoretical guarantees for the trigger to limit learning experiments to the necessary instants. Event-triggered learning improves robustness towards changing environments and reduces communication even further then classical ETSE, which is demonstrated in numerical and physical experiments on a cart-pole system.

\end{abstract}

%% file: texs/1introduction.tex
\section{Introduction}
Networked control systems (NCSs) are rapidly gaining in popularity, both in academia and industry.
Advancements in control strategies and network technologies enable the systems to closely interact with their environment and share data. 
Treating communication as a shared resource, as suggested in \cite{HespanhaNaghshtabriziXuJan07}, is an important step to scale NCSs to problems involving many agents.

In this paper, we consider NCSs with multiple spatially distributed agents, whose dynamics are independent, but that are coupled through a joint control objective and communicate via a shared network.
Figure \ref{multiagent} depicts two agents representative for one communication link in such an NCS. 
While communication between agents is beneficial or even necessary for coordination (\eg formation control \cite{cao2013overview}, or multi-agent balancing \cite{TrDAn12b}), the network constitutes a shared and scarce resource and, hence, its usage shall be limited.
% to the necessary instants.

Event-triggered state estimation (ETSE) \cite{YoTiSo02,TrDAn11,sijs2012event,wu2013event, TAC14b_web}
% removed:
%  MuTr18, tr17
has been proposed to reliably exchange sensor or state data between agents, but with limited inter-agent communication.
Many ETSE methods utilize dynamics models to predict other agents' states or measurements (see \fig \ref{multiagent}), in order to anticipate their behavior without the need for continuous data transmissions.
% in optimal decision making.
Event triggering rules then ensure that an update is sent whenever the prediction is not good enough.  Hence, the capability of the dynamics model in making accurate predictions critically determines how much communication can be saved.
%However, most of them rely on precise models of the underlying dynamical system.

In order to improve prediction accuracy, we propose to combine ETSE with model learning.  We develop a framework, where an agent uses its input-output data to update 
%(or re-compute) 
the model of its dynamics,
% whenever prediction performance is poor, 
and communicates this model to other agents.
% in order for them to improve their predictions.  
 With improved models, it is possible to achieve superior prediction accuracy and thus lower communication even further compared to standard ETSE, especially when facing changing dynamics. 
Since learning and communicating models can be resource-intensive operations themselves, 
%costly actions, 
%we also limit model learning to the necessary instants.
we develop triggering rules that quantify the model accuracy and decide when to learn.  The result is an \emph{event-triggered learning} scheme, which sits on top of the standard ETSE framework  (\cf \fig \ref{multiagent}).
%event-triggered learning to quantify the model accuracy and decide when to trigger learning experiments. 

\begin{figure}[t]
	\begin{center}
		\includegraphics[width=0.48\textwidth]{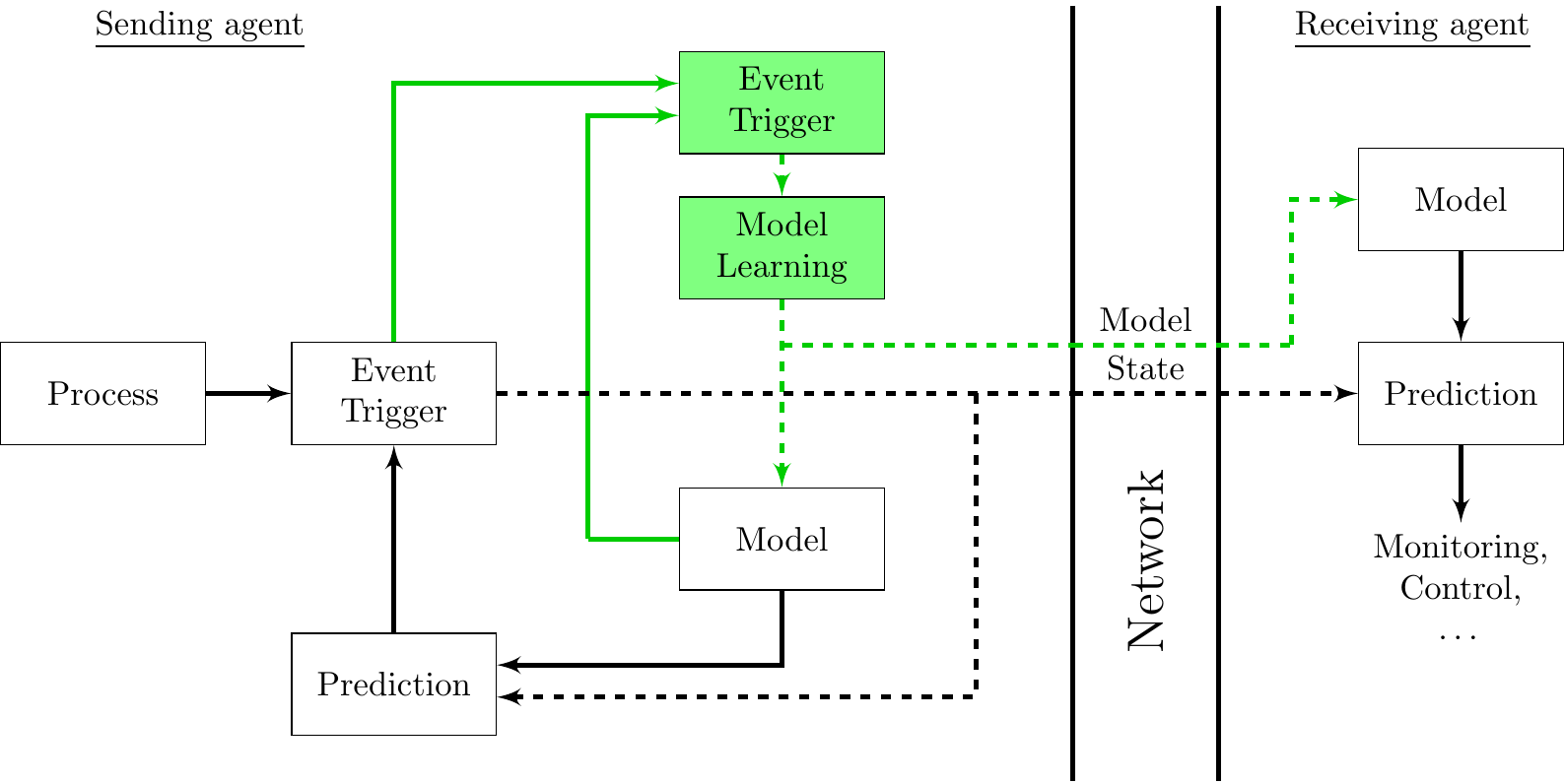}
		\vspace{-4mm}
		\caption{Proposed event-triggered learning architecture. Based on a typical event-triggered state estimation architecture (in black), we propose model learning (in green) to improve predictions and thus lower communication between Sending and Receiving agents. Learning experiments are triggered by comparing empirical with expected inter-communication times.}
		\label{multiagent}
	\end{center}
\vspace{-3mm}
\end{figure}
\subsubsection*{Contributions}
In detail, this paper makes the following contributions:
\begin{itemize}
	\item introducing the novel idea of event-triggered learning;
	\item derivation of a data-driven learning trigger based on statistical properties of inter-communication times;
	\item probabilistic guarantees ensuring the effectiveness of the proposed learning triggers;
	\item concrete implementation of event-triggered learning for linear Gaussian dynamic systems; and
	\item demonstration of improved communication behavior in numerical simulations and hardware experiments on a cart-pole system.
\end{itemize}

\subsubsection*{Related work}
Various event-triggering approaches have been proposed for improving resource usage in state estimation; for recent overviews see, \eg \cite{shi2015event,trimpeEBCCSP15,SiNoLaHa16,tr17} and references therein.  
Approaches differ, among others, in the type of models that are used for predictions.  
The send-on-delta protocol \cite{miskowicz2006send} assumes a trivial constant model by triggering when the difference between the current and lastly communicated value passes a threshold.  This protocol is extended to linear predictions in \cite{Su07}, which are obtained from approximating the signal derivative from data.  More elaborate protocols use dynamics models of the observed process, which typically leads to more effective triggering \cite{SiKeNo14,trimpeEBCCSP15}.  The vast majority of ETSE methods (\eg \cite{YoTiSo02,TrDAn11,sijs2012event,wu2013event, TAC14b_web,shi2015event,trimpeEBCCSP15,SiNoLaHa16,tr17,MuTr18,SiKeNo14}) employ linear dynamics models, which we also consider herein. Nonlinear prediction models are used in \cite{TrBu15,MaEsGaSa15}, for example. None of these works considers model learning to improve prediction accuracy, as we do herein. While we implement event-triggered learning for linear Gaussian problems, the general idea also applies to other types of estimation problems and dynamics models.

In order to obtain effective learning triggers, we take a probabilistic view on inter-communication times (\ie the time between two communication events) and trigger learning experiments whenever the expected communication differs from the empirical.  A similar interpretation of inter-communication times is considered in \cite{XuHe04}, where it is used to infer stability results. We extend this idea with statistical convergence results to design the event trigger for learning.  

% Prev version:
%Our approach is based on a probabilistic point of view to quantify model accuracy. We derive how expected communication behavior should look like and quantify the probability of deviations from the expected behavior. If we observe repeatedly very unlikely events, we conclude structural problems in the model and trigger learning. A similar probabilistic interpretation of intercommunication times is considered in \cite{XuHe04}, where it is used to analyze communication behavior for the purpose of stability analysis. We extend this idea with statistical convergence results and use it to design the event trigger for learning. 

Deciding if learning is necessary is related to the problem of change detection, which seeks to identify times when the distribution governing a random process changes.  Change detection has received considerable attention in statistical literature; see \cite{lai1995sequential} for an overview. 
While one main application of this work
%A main application of change detection focuses on 
is fault detection,
% in order to prevent damage or react accordingly to the change.  
%However, our method is structurally different, since 
we use the proposed trigger to initiate learning experiments to capture the changed distribution.
% experiments and thus reduce communication in the long-term.
Furthermore, we do not analyze the stochastic process directly, but instead target inter-communication times, which are a natural one-dimensional feature in the ETSE framework and also amenable to theoretical analysis, which is incorporated in the trigger design.
Applications of change detection in the context of NCSs focus on fault detection in the presence of delays \cite{wang2008new}, or network design and performance in general \cite{haghighi2013dynamic}.
Iterative learning control has been proposed for multi-agent problems in \cite{zhang2016event} to improve event-triggered control for repetitive tasks, however,
the idea of triggering learning when there is significant change is new to the best of the authors' knowledge.

%There are contributions applying change point detection methods to NCS, but focus on delays \cite{wang2008new} or the network design and performance in general \cite{haghighi2013dynamic}. \hl{The idea of structured model learning in a ETSE framework is to the best of our knowledge new.}

\section{Event-triggered Learning}
In this section, we explain the main idea of event-triggered learning using the schematic in \fig \ref{multiagent}. The figure depicts a canonical problem, where one agent (`Sending agent' on the left) has information that is relevant for another agent at a different location (`Receiving agent'). This setting is representative for remote monitoring scenarios or two agents of a multi-agent network, for instance. For resource-efficient communication, a standard ETSE architecture
% Removed citations here, because we discuss ETSE already in related work
% (see \cite{tr17} or \cite{lemmon2010event}) 
is used (shown in black). The main contribution of this work is to incorporate learning into the ETSE architecture. By designing an event trigger also for model learning (in green), learning tasks are performed only when necessary.
%we limit the learning part to the necessary instants. 
We next explain the core components of the proposed framework.

The sending agent in 
\fig \ref{multiagent}
%Fig. 1 
monitors the state of a dynamic process (either
directly measured or obtained via state estimation) and can transmit this state
information to the remote agent via a network link. In order to save network
resources, an event-based protocol is used. The receiving agent makes 
model-based predictions (`Model' and 'Prediction') to anticipate the state at times
of no communication.  The sending agent implements a copy of the same prediction
and compares it to the current state in the `Event Trigger', which triggers a communication whenever the prediction deviates too much from the actual state.
This general scheme is standard in ETSE literature (see \cite{shi2015event,trimpeEBCCSP15,SiNoLaHa16,tr17} and references therein). 
%While different triggering laws can be used for this decision, 
The effectiveness of this scheme in reducing communication will generally depend on the accuracy of the prediction, and thus the quality of the model.

The key idea of this work is to trigger learning experiments and learn improved models from data when prediction performance of the current model is poor.
%Either when needed, or online in parallel, the data acquired about the process can be used to learn a new dynamics model. 
The updated model is then shared with the remote agent to improve its predictions. Because performing a learning task is costly itself (\eg involving computation and communication resources, as well as possibly causing deviation from the control objective),
%(\eg computational cost, deviation from the control objective, as well as communication cost in this setting), 
we 
%implement learning on demand and 
propose event-triggering rules also for model learning.

While the idea of using event triggering to save communication in estimation or control is quite common by now, this work proposes event triggering also on a higher level. Triggering of learning tasks yields improved prediction models,
%We trigger the learning of new models, 
which are the basis for ETSE at the lower level.

The general idea of event-triggered learning 
%for improving model accuracy on demand 
applies to diverse scenarios.  In the following, we make the idea concrete
%precise and develop event-triggered learning specifically 
for ETSE of linear Gaussian systems (introduced in Sec. III). For this case, model learning can be solved with standard least-squares estimation (Sec. IV). We propose a trigger design that is based on probabilistic analysis of inter-communication times (Sec. V and VI). By means of numerical examples and physical experiments on a cart-pole system (Sec. VII and VIII), we show that communication can effectively be reduced through event-triggered learning.

%% file: texs/2preliminaries.tex
\section{Linear Gaussian Problem}
\label{sec:problem}
In this section, we make the problem of event-triggered learning precise for
%state the problem for 
linear Gaussian dynamic systems. 
For this, we introduce all standard elements of the ETSE architecture shown in \fig \ref{multiagent} (black).
%We introduce all standard elements in ETSE (Fig. 1, black) and motivate the necessity of accurate models. Similar architectures can be found e.\,g. in \cite{HespanhaNaghshtabriziXuJan07}, \cite{lemmon2010event} and references therein.
\subsection{Process}
We assume the linear dynamics (`Process' in \fig \ref{multiagent})
\begin{equation}
	x(k+1) = A x(k) + Bu(k) + \epsilon(k),
	\label{LTI}
\end{equation}
with discrete-time index $k$, state $x(k) \in\mathbb{R}^n$, control input $u(k) \in \mathbb{R}^q$, system matrices $A \in \mathbb{R}^{n\times n}$, $B \in \mathbb{R}^{n\times q}$, and independent identically distributed (i.i.d.) Gaussian random variables $\epsilon(k) \sim \mathcal{N}(0, \Sigma)$.  For simplicity, we assume the state $x(k)$ can be measured, but the framework readily extends to the case where the state is locally estimated.
% from measurements.

The system $(A,B)$ is assumed stabilizable. Hence, stable closed loop dynamics can be achieved through state feedback
\begin{equation}
 u(k) = Fx(k) + r(k),
 \label{eq:stateFeedback}
\end{equation}
where $r(k)$ is a known reference. 
The reference can be used for tracking problems or to excite the system, which is particularly important during learning experiments and will be further discussed later.
The closed loop dynamics then reads
\begin{equation}
	x(k+1) = A_{\mathrm{CL}} x(k) + Br(k) + \epsilon(k),
	\label{CL}
\end{equation}
with stable matrix $A_{\mathrm{CL}} \df (A+BF)$.
%\begin{equation}
%	A_{\mathrm{CL}} \df (A+BF).
%\end{equation}
% Comment by ST: Make this an inline equations; allows you to save space and you don't need this equation to be referenced.

\subsection{Predictions}
\label{sec:ETSE_pred}
The `Prediction' block, which is implemented on the sending and the receiving agent (\cf \fig \ref{multiagent}), utilizes a model $(\hat{A},\hat{B})$ to predict the true process. 
%This is done by both, the sending and the receiving agent. 
After initializing $\hat{x}(0) = x(0)$, both agents iterate
\begin{equation}
 \hat{x}(k+1) = \hat{A}_{\mathrm{CL}} \hat{x}(k) + \hat{B} r(k).
 \label{OLPred}
\end{equation}
The prediction \eqref{OLPred} is deterministic and deviates from the true value \eqref{CL} over time.
We bound this error through an event trigger: whenever a certain error threshold is reached, the actual state $x(k+1)$ is transmitted and the prediction reset
\begin{equation}
\hat{x}(k+1) = \begin{cases} \hat{A}_{\mathrm{CL}} \hat{x}(k) + \hat{B} r(k) &\mbox{\text{if} } \gamma_{\text{state}}=0\\
															 x(k+1)	&\mbox{\text{if }}  \gamma_{\text{state}}=1 \end{cases},
\label{predwithcom}
\end{equation}
where the binary variable $\gamma_{\text{state}}$ denotes positive ($\gamma_{\text{state}}=1$) or negative triggering decisions.

%This way, communication is reduced to the necessary instants which have a significant prediction error.

\subsection{Event Trigger}
\label{sec:ETSE_trig}
The `Event Trigger' on the sending agent has access to both, prediction and true state. It can thus track the error 
\begin{equation}
	\|z(k)\|_2 \df \| x(k) - \hat{x}(k) \|_2
	\label{predictionerror}
\end{equation}
and trigger a communication when it becomes too large:
\begin{equation}
%\text{communicate }x(k) \iff 
\gamma_{\text{state}} = 1 \iff \|z(k)\|_2 \geq \delta.
\label{eq:est_trigger}
\end{equation}
%\begin{equation}
%\text{communicate }x(k) \iff \gamma_{\text{state}} = 1.
%\end{equation}
% ST: The communication rule is not really interesting I think...  What you want to write down is the triggering rule...
%
This way, communication is reduced to the necessary instants given by a significant prediction error.

\subsection{Problem Formulation}
Precise models $(\hat{A},\hat{B},\hat{\Sigma})$ are key to reduce communication rates.
In this work, we address mismatches between model $(\hat{A},\hat{B},\hat{\Sigma})$ and true process $(A, B, \Sigma)$, which may stem from imprecise initial models or changing dynamics, for example.  
The development of a model learning framework to improve the effectiveness of ETSE in reducing communication is the main objective of this paper.
%The objective is to develop an effective model learning framework to improve the effectiveness of ETSE in reducing communication.  
This includes, in particular, a decision rule (the \emph{learning trigger}  $\gamma_{\text{learn}}$) for deciding when a new model $(\hat{A},\hat{B},\hat{\Sigma})$ ought to be learned.

% In this work, we address a mismatch between model and true process. We develop an event trigger $\gamma_{\text{learn}}$ for model learning based on the accuracy of long term predictions (Sec. V and VI). After detecting a mismatch ($\gamma_{\text{learn}} = 1$), we resolve the issue by learning improved models (Sec. IV) and transmitting these to all relevant parts of the scheme. This way, we obtain improved prediction performance and thus lower communication rates. 

%% file: texs/3learning.tex
\section{Model Learning}
\label{sec:modelLearning}
This section addresses learning of the dynamics model~\eqref{OLPred} from input-output data of the system \eqref{CL} ('Model Learning' in \fig \ref{multiagent}). In the numerical and physical experiments of \sect \ref{sec:simulation} and \ref{sec:experiment}, the data is obtained from dedicated learning experiments.  In other settings, data could be recorded during normal operation.  We here present standard least squares to estimate the system matrices, but any other technique for linear system identification \cite{ljung1999system} could be used.
 
Rewriting \eqref{CL} as
\begin{equation}
	x(k+1) = \begin{bmatrix} A_{\mathrm{CL}}  & B \end{bmatrix} 
	\begin{bmatrix} x(k) \\ r(k) \end{bmatrix} + \epsilon(k)
\end{equation}
and stacking $M$ samples yields 
\begin{equation}
	Y = \begin{bmatrix} A_{\mathrm{CL}} & B \end{bmatrix} X + \begin{bmatrix}
\epsilon(0) &\ldots &\epsilon(M-1)  
\end{bmatrix},
\end{equation}
with
\begin{equation}
X \df \begin{bmatrix}
x(0) \! & \! \ldots \! & \! x(M-1) \\
r(0) \! & \! \ldots \! & \! r(M-1)
\end{bmatrix},~
Y \df \begin{bmatrix}
x(1) \! & \! \ldots \! & \! x(M)
\end{bmatrix} 
\end{equation}

Thus, the ordinary least squares (OLS) estimator for the model matrices yields
$
	[\hat{A}_{\mathrm{CL}} ~\hat{B}] = YX^\intercal (X X^\intercal)^{-1}.
$

% Prev. version:
%
%We rewrite the system \eqref{CL} to separate the model matrices $A_{\mathrm{CL}}$ and $B$ from the reference input signal $r(k)$ and state $x(k)$.
%This yields
%\begin{equation}
%	x(k+1) = [A_{\mathrm{CL}} ~B] \left[x(k) ~ r(k)\right]^\intercal + \epsilon(k).
%\end{equation}
%In order to express multiple iterations in one equation we define the data matrix
%\begin{equation}
%X\df \left[
%\begin{matrix}
%x(0)  & \ldots & x(M-1) \\
%r(0)  & \ldots & r(M-1)
%\end{matrix}
%\right],
%\end{equation}
%the output 
%\begin{equation}
%	Y \df \left[ \begin{matrix}
%x(1) & \ldots & x(M)
%\end{matrix}\right],
%\end{equation}
%and process noise
%\begin{equation}
%E \df 	\left[ \begin{matrix}
%\epsilon(0) &\ldots &\epsilon(M-1)  
%\end{matrix}\right].
%\label{processnoise}
%\end{equation}
%This way we obtain
%\begin{equation}
%	Y = \left[A_{\mathrm{CL}} ~B\right] X + E
%\end{equation}
%and thus the ordinary least squares (OLS) estimator for the model matrices yields
%\begin{equation}
%	[\hat{A}_{\mathrm{CL}} ~\hat{B}] = YX^\intercal (X X^\intercal)^{-1}.
%	\label{OLSest}
%\end{equation}

Due to the auto-regressive structure of the system \eqref{CL}, we can ensure identifiability for certain types of input signals $r$. Sinusoidal input signals with increasing frequencies (chirp) are known to yield invertible matrices $X X^\intercal$ (see condition of persistent excitation, \eg \cite{ljung1999system}). 
We thus use chirp signals as reference $r$ to excite the system when a learning experiment shall be performed.
Since the estimator is designed to achieve minimal variance it is straightforward to obtain $\hat{\Sigma}$ by averaging over the squared residuals.
% in the examples in sections \ref{sec:simulation} and \ref{sec:experiment}.
%Whenever $\gamma_{\mathrm{learn}} = 1$, we start a learning experiment and assign such chirp signals as the reference signal $r(k)$ to excite the system.
%\commentseb{I suggest to delete the next paragraph, because I already have in the introduction now that other methods can be used.  Also, PEM has no special relevance for us, no?}
%There are also other sensible approaches for identifying the system efficiently. For example prediction-error minimization (PEM) methods (see \cite{ljung1999system} for more details) rely on gradient descent and minimize the prediction-error starting with an initial model guess. Depending on the initial guess PEM methods can be powerful alternatives to least square methods.

The specific learning method only affects the 'Model Learning' block in \fig \ref{multiagent} and has no influence on the decision whether learning is necessary ('Event Trigger'). Hence, the general event-triggered learning approach is agnostic to what learning or identification method is used.
% as long as it yields an improved model. 

%% file: texs/4ETL.tex
\section{Foundations in Stochastic Analysis}
In this section, we briefly discuss theoretical results from stochastic analysis as background for deriving the event trigger for model learning in the next section.
% (green 'Event Trigger' block in Fig. 1). 
More detailed treatments of stochastic processes and stochastic differential equations (SDEs) can be found in \cite{oksendal2003stochastic} and \cite{Sc09}, for example.

Inter-communication times (\ie the time between two triggering events \eqref{eq:est_trigger}) will play a key role in the proposed triggering scheme for model learning.
%One of the key ideas behind the triggering scheme for model learning is the analysis of intercommunication times from a probabilistic point of view. 
In the following, we express inter-communication times as random variables (stopping times) and deduce statistical properties directly from the model (in \sect \ref{sec:stoppingTimes}).  To be amenable to stopping time analysis, we shall first describe the event-triggered estimation scheme in terms of continuous-time SDEs (\sect \ref{sec:OUprocess}).

The analysis in this section is done under the assumption that the prediction model \eqref{OLPred} is perfect (\ie $\hat{A}_\mathrm{CL} = A_\mathrm{CL}$, $\hat{B} = B$, and $\hat{\Sigma}=\Sigma$), which will be motivated in the next section.

%that the process is perfectly described by the model. 
%This assumption is essential to keep the expected behavior of intercommunication times tractable. Furthermore, it yields model based reference values, which we use to develope
%the actual trigger in Sec. VI. 

\subsection{Ornstein-Uhlenbeck processes}
\label{sec:OUprocess}
Consider \eqref{CL} with $r = 0$, \ie
%Assume the discrete time system dynamics
\begin{equation}
x(k+1) = A_{\mathrm{CL}}x(k) + \epsilon(k) .
\label{disctimesys}
\end{equation} 
%which is \eqref{CL} with the reference signal $r(k)$ equal to zero. 
Because of the perfect model assumption, the reference signal $r(k)$ cancels out in the later analysis step \eqref{ctstimepred} and is hence irrelevant.
Transforming system \eqref{disctimesys} to continuous time yields an Ornstein-Uhlenbeck (OU) process
\begin{equation} 
	\mathrm{d}X(t) = \mathcal{A} X(t) \mathrm{d}t+ Q \mathrm{d}W(t),
	\label{OU}
\end{equation}
with $t \in \mathbb{R}$, $X(t) \in \mathbb{R}^n$, and $W(t)\in \mathbb{R}^n$ a multidimensional Wiener process. The matrices $\mathcal{A}$, $Q \in \mathbb{R}^{n\times n}$ can be computed from $A_{\mathrm{CL}}$, the sampling time of the discrete-time system \eqref{disctimesys}, and the covariance $\Sigma$ (refer to \cite{Sc09,aastrom2012introduction} for details).
Continuous- and discrete-time processes coincide in the sense that they have the same distribution for any given sampling time.

The OU process \eqref{OU} can be written explicitly in integrated form as
\begin{equation}
 X(t) = e^{\mathcal{A}t}X(0) + \int_0^t e^{\mathcal{A}(t-s)}Q \, \mathrm{d}W(s).
\label{intOU}
\end{equation}
With the perfect model assumption,
%(i.e.\ $\mathcal{A} = \hat{\mathcal{A}}$),
the discrete-time predictions \eqref{OLPred} transform to 
%the deterministic process
\begin{equation} 
\hat{X}(t) = e^{\mathcal{A}t}X(0)
\label{ctspred}
\end{equation}
in continuous time.
Using \eqref{intOU} and \eqref{ctspred}, we thus obtain for the continuous-time prediction error 
\begin{equation}
	Z(t) =  X(t) - \hat{X}(t)  = \int_0^t e^{\mathcal{A}(t-s)}Q \mathrm{d}W(s).
	\label{ctstimepred}
\end{equation}
By comparing to \eqref{intOU}, it follows that $Z(t)$ is an OU process starting in zero. We will further analyze the behavior of $\|Z(t)\|_2$ and make a connection between stopping times of this process and communication behavior.

\subsection{Stopping Times}
\label{sec:stoppingTimes}
Due to the existence and uniqueness result for SDEs \cite{oksendal2003stochastic}, continuity of sample paths is assured for OU processes. 
We will leverage this property to pinpoint the exact moment the process crosses a given threshold.
Exactly like in the discrete-time setting, state information is communicated whenever $\|Z(t)\|_2 \geq \delta$,
which resets $Z(t)$ to $0$.
Accordingly, the stopping time $\tau$ is defined as the first time when the stochastic process $Z(t)$ exits an $n$-dimensional sphere with radius $\delta$:
\begin{equation}
\tau \df \inf \{ t : \| Z(t) \|_2 \geq \delta \}.
\end{equation}
The random variable $\tau$ hence corresponds to the inter-communication time, the time between communication events.
%, which we also call intercommunication times. 
Because they coincide in the setting herein, `inter-communication time' and `stopping time' will be used synonymously hereafter.
In \sect \ref{sec:eventTriggerLearning}, learning triggers will be proposed that are based on a comparison of empirical and expected inter-communication times. The computation of expected inter-communication times $\mathbb{E}[\tau|Z(0)\!=\!0]$ is discussed next.
%We propose to investigate $\mathbb{E}[\tau]$ and to design learning triggers based on this analysis. 

For certain classes of SDEs, it is possible to compute $\mathbb{E}[\tau|Z(0)\!=\!x]$ as the solution $v(x)$ to specific partial differential equations (PDEs).
The following lemma states this result 
%We will discuss this result 
for the OU process.
% in the following lemma and state a special case where we can give an explicit solution $v(x)$.
\begin{lem}
Assume the boundary value problem
\begin{equation}
\begin{alignedat}{2}
\frac{\partial v(x)}{\partial x}\mathcal{A} x + \frac{1}{2}\mathrm{trace}\left[Q^{\intercal} \frac{\partial^2 v(x)}{\partial x^2} Q  \right] &= -1 \quad &&\mbox{in } \Omega,\\
\label{PDE}v(x) &= 0  &&\mbox{on } \partial\Omega,
\end{alignedat}
\end{equation}
with gradient $\frac{\partial v(x)}{\partial x}$, Hessian $\frac{\partial^2 v(x)}{\partial x^2}$, and $\Omega$ an $n$-dimensional sphere with radius $\delta$, 
has a unique bounded solution $v(x)$. 
%Here, $\frac{\partial v(x)}{\partial x}$ is the gradient vector, $\frac{\partial^2 v(x)}{\partial x^2}$ is the Hessian matrix and $\Omega$ is a $n$-dimensional sphere with radius $\delta$.
Then, $\mathbb{E}[\tau|Z(0)\!=\!0] = v(0)$.
% the expected value of the stopping time is equal to the solution to the PDE evaluated in zero: 
%\begin{equation}	
%	\mathbb{E}\left[\tau~|~Z(0)=0\right] = v(0).
%\end{equation}
\end{lem}
\begin{proof}
The result is obtained by using the more general Andronov-Vitt-Pontryagin formula from \cite{Sc09} and adapting it to the OU process (as was also done in \cite{XuHe04}).
\end{proof}

The one-dimensional case can be solved analytically (see Example 4.2 on page 111 in \cite{Sc09}) and gives interesting insights. 
From the solution, dependencies between $\mathbb{E}[\tau|Z(0)\!=\!x]$ and the parameters $\mathcal{A}$ and $Q$ can be seen. The magnitude of stochastic effects $Q$ is pushing the process out of the domain, while the stable part $\mathcal{A}$ drives it back to zero. Hence, more stable $\mathcal{A}$ leads to longer stopping times, while larger $Q$ leads to shorter.
% Prev version:
%Following \cite{graczyk2008exit}, the solution to \eqref{PDE} can also be expressed in terms of hypergeometric functions ${}_2F_2$ which are implemented in most numerical libraries and can be sampled efficiently. We obtain 
%\begin{equation} 
%\mathbb{E}[\tau~|~Z(0)=0] = \left(\frac{\delta^2}{Q}\right) {}_2F_2\left(1,1;\frac{3}{2},2;\mathcal{A}\delta^2\right).
%\end{equation}

For general dimension $n$, one could attempt solving the PDE \eqref{PDE}, which is, however, challenging because it is nonlinear and possibly high dimensional.
Typically, one has to resort to numerical solutions, which usually yield the whole function $v(x)$ and are computationally intensive. 
Since we actually only require $v(0)$, an alternative is to approximate $\mathbb{E}[\tau|Z(0)\!=\!0]$ through Monte Carlo simulations, which we use in the experiments herein.  
Because the error $Z(t)$ is always reset to $0$ at triggering instants for the scenario herein, we shall omit the dependence on the initial value and write  $\mathbb{E}[\tau]$ instead of $\mathbb{E}[\tau|Z(0)\!=\!0]$ in the following.

%% file: texs/5trigger.tex
\section{Event Trigger for Model Learning}
\label{sec:eventTriggerLearning}
In this section, we design the event trigger for model learning (green 'Event Trigger' block in \fig \ref{multiagent}).
% (shown in Fig. 1 in green). 
We first discuss the general idea of how to trigger learning experiments, and then state two concrete instances.
%While in the previous section we derived how to compute $\mathbb{E}[\tau]$ analytically, we will now adopt a statistical point of view. 

\subsection{General idea}
The learning trigger must reliably detect when the prediction accuracy of the current model is poor, and thus learning a new model from data promises improved predictions.
We cannot directly compare the model $(\hat{A}_\mathrm{CL}, \hat{B}, \hat{\Sigma})$ to the true process $(A_\mathrm{CL}, B, \Sigma)$ because it is unknown. 
Owing to the stochasticity of the process, it is also not straightforward to quantify model accuracy from data. 
Since we ultimately care about the communication behavior that results from the models, we propose to utilize inter-communication times to trigger learning in the following way.

If an accurate model $(\hat{A}_\mathrm{CL}, \hat{B}, \hat{\Sigma})$ is used, average inter-communication times are expected to be equal to $\mathbb{E}[\tau]$ (as computed according to \sect \ref{sec:stoppingTimes} under the assumption $(\hat{A}_\mathrm{CL}, \hat{B}, \hat{\Sigma}) = (A_\mathrm{CL}, B, \Sigma)$).  If, however, observed inter-communication times deviate on average from $\mathbb{E}[\tau]$, we conclude an inaccurate model and trigger a learning experiment ($\gamma_{\mathrm{learn}}\! =\! 1$).  Hence, the proposed learning trigger compares empirically observed inter-communication times with the expected inter-communication time computed from the current model, and triggers learning when the two are significantly apart.

% prev version:
%Assuming a given model, we discussed how to compute $\mathbb{E}[\tau]$ (see Sec. V).
%The classical ETSE scheme (shown in Fig. 1 in black) yields empirical intercommunication times $\tau_1,\ldots,\tau_N$, which we average.
%By the law of large numbers, the empirical mean converges to the expected value. 
%Furthermore we can quantify the speed of convergence with the aid of Hoeffding's inequality. 
%Based on the difference between the analytically derived expected value and the observed empirical mean, we design the trigger for model learning.
%If the difference is significantly higher then what is expected, we conclude an inaccurate model and set $\gamma_{\mathrm{learn}} = 1$. 

\subsection{Exact learning trigger}
Based on the foregoing discussion, we propose the following learning trigger:
%We propose the following trigger to detect inaccurate models
\begin{equation}
\gamma_{\mathrm{learn}} = 1 \iff	\left| \frac{1}{N}\sum\limits_{i=1}^N \tau_i - \mathbb{E}[\tau] \right| \geq \kappa_{\mathrm{exact}},
	\label{trigger}
\end{equation}
where $\gamma_{\mathrm{learn}} = 1$ indicates that a new model shall be learned; $\mathbb{E}[\tau]$ is computed according to \sect \ref{sec:stoppingTimes} assuming a perfect model; and $\tau_1, \tau_2, \dots, \tau_N$ are the last $N$ empirically observed inter-communication times ($\tau_i$ the duration between two triggers \eqref{eq:est_trigger}). The moving horizon $N$ is chosen to yield robust triggers.
The threshold parameter $\kappa_{\mathrm{exact}}$ quantifies the error we are willing to tolerate. 
We denote \eqref{trigger} as the \emph{exact learning trigger} because it involves the exact expected value $\mathbb{E}[\tau]$, as opposed to the trigger derived in the next subsection.

% Prev version:
%For the experiments we will use a moving average of size $N$ and additionally enforce the triggering rule to hold for a certain duration. The moving average is necessary to detect changes in the system, while the time requirement addresses the multiple comparisons problem and adds robustness to the trigger. 

Even though the trigger \eqref{trigger} is meant to detect inaccurate models, there is always a chance that the trigger fires not due to an inaccurate model, but instead due to the randomness of the process (and thus randomness of inter-communication times $\tau_i$).  
Even for a perfect model, \eqref{trigger} may trigger at some point.  This is inevitable due to the stochastic nature of the problem.  
%In order to obtain effective triggers, 
Therefore, we propose to chose $\kappa_{\mathrm{exact}}$ to yield effective triggering with a user-defined confidence level.  For this, we make use of Hoeffding's inequality:
%
%We propose to chose the parameter $\kappa_{\mathrm{exact}}$ in a structured way \commentseb{CAn you be more specific here?  Structured way isn't really clear...}, based on Hoeffding's inequality.
%\begin{lem}[Strong law of large numbers] Assume i.\,i.\,d. random variables $\tau_1,\tau_2,\ldots$ with $\mathbb{E}[\tau_1]=\mathbb{E}[\tau_2] = \ldots = \mu < \infty$. Then
%$$ \frac{1}{N}\sum\limits_{i=1}^N \tau_i \xrightarrow[]{N \to \infty} \mu \quad \text{almost surely}.$$
%\end{lem}
\begin{lem}[Hoeffding's inequality \cite{von2008statistical}]
Let $\tau_1,\ldots,\tau_N$ be i.i.d.\ bounded random variables with $\tau_i \in [0,\tau_{\mathrm{max}}]$ for all $i$. Then
\end{lem}
\begin{equation} 
	\mathbb{P}\left[\left| \frac{1}{N}\sum\limits_{i=1}^N \tau_i -\mathbb{E}[\tau] \right|\geq \kappa \right] \leq 2e^{-\frac{2N\kappa^2}{\tau^2_{\mathrm{max}}}}.
\end{equation}

The application of Hoeffding's inequality requires inter-communication times $\tau_i$ to be upper bounded by some $\tau_{\mathrm{max}}$.  This is easily enforced in practice by triggering at the latest when $\tau_{\mathrm{max}}$ is reached.  Moreover, we specify a maximal probability $\eta$ that we are willing to tolerate for the difference being caused through randomness.  We then have the following result for the trigger \eqref{trigger}:
%
% prev version:
%The assumption  $\tau \in [0,\tau_{\mathrm{max}}]$ enforces communication after time $\tau_{\mathrm{max}}$. Considering physical systems, this assumption is mostly fulfilled due to safety concern.
%We define a maximal probability $\eta$ we are willing to tolerate for the difference being caused through randomness.  
%If we observe high deviations, we conclude with high probability an inaccurate model and trigger model learning. 
%
%Using the parameter $\eta$ and Hoeffding's inequality, we can calculate how to choose $\kappa$ in order to achieve the given precision.
\begin{theo}[Exact learning trigger] 
\label{thm:exactTrigger}
Let the parameters $\eta$, $N$, and $\tau_{\mathrm{max}}$ be given, and assume a perfect model  $(\hat{A}_\mathrm{CL}, \hat{B}, \hat{\Sigma}) = (A_\mathrm{CL}, B, \Sigma)$. 
%If $\kappa_{\mathrm{exact}}$ is chosen as 
For
\begin{equation}
\kappa_{\mathrm{exact}} = \tau_{\mathrm{max}}\sqrt{-\frac{1}{2N} \ln \frac{\eta}{2} },
\label{eq:kappa_exact}
\end{equation}
we have
\begin{equation}
	\mathbb{P}\left[\left| \frac{1}{N}\sum\limits_{i=1}^N \tau_i -\mathbb{E}[\tau] \right|\geq \kappa_{\mathrm{exact}} \right] \leq \eta.
\end{equation}
\end{theo}
\vspace{1ex} % ST: hack, don't know why needed...
\begin{proof}
	Substituting \eqref{eq:kappa_exact} for $\kappa_{\mathrm{exact}}$ into the right hand side of Hoeffding's inequality yields the desired result.
\end{proof}

The theorem quantifies the expected convergence rate of the empirical mean to the expected value for a perfect model.
This result can be used as follows: the user specifies the desired confidence level $\eta$, the number $N$ of inter-communication times considered in the empirical average, and the maximum inter-communication time $\tau_{\mathrm{max}}$.  By choosing $\kappa_{\mathrm{exact}}$ as in \eqref{eq:kappa_exact}, the exact learning trigger \eqref{trigger} is guaranteed to make incorrect triggering decisions (false positives) with a probability of less than $\eta$.

\subsection{Approximated learning trigger}
As discussed in \sect \ref{sec:stoppingTimes}, obtaining $\mathbb{E}[\tau]$ can be difficult and computationally expensive. Instead of solving the PDE \eqref{PDE}, we propose to approximate $\mathbb{E}[\tau]$ by sampling $\tau$. For this, we simulate the OU process $Z(t)$ \eqref{ctstimepred} until it reaches a sphere with radius $\delta$ for $M$ times, and average the obtained stopping times $\tau^{\text{sim}}_1,\ldots, \tau^{\text{sim}}_M$. Hence,
\begin{equation}
	\mathbb{E}[\tau] \approx \frac{1}{M}\sum\limits_{i=1}^M \tau^{\text{sim}}_i.
	\label{mean}
\end{equation}
This yields the \emph{approximated learning trigger}
\begin{equation}
\gamma_{\mathrm{learn}} = 1 \iff	\left| \frac{1}{N}\sum\limits_{i=1}^N \tau_i - \frac{1}{M}\sum\limits_{i=1}^M \tau^{\text{sim}}_i \right| \geq \kappa_{\mathrm{approx}}.
\label{approxtrigger}
\end{equation}

This approximation leads to a different choice of $\kappa_{\mathrm{approx}}$.
\begin{theo}[Approximated learning trigger] 
\label{thm:approxTrigger}
Let the parameters $\eta$, $N$, $M>N$, and $\tau_{\mathrm{max}}$ be given, and assume $(\hat{A}_\mathrm{CL}, \hat{B}, \hat{\Sigma}) = (A_\mathrm{CL}, B, \Sigma)$. 
%If $\kappa_{\mathrm{approx}}$ is chosen as 
For
\begin{equation}
\kappa_{\mathrm{approx}} = \tau_{\mathrm{max}}\sqrt{-\frac{2}{N} \ln \frac{\eta}{4} },
\end{equation}
we have
\begin{equation}
		 \mathbb{P}\left[\left|\frac{1}{N}\sum\limits_{i=1}^N \tau_i - \frac{1}{M}\sum\limits_{i=1}^M \tau^{\text{sim}}_i \right| \geq \kappa_{\mathrm{approx}} \right] < \eta.
\end{equation}
\end{theo}
\vspace{1ex}

\begin{proof}
We insert $\mathbb{E}[\tau]$, use the triangle inequality, and additivity of probability measures:
\begin{align*}
&{}	&&\mathbb{P}\left[\left|\frac{1}{N}\sum\limits_{i=1}^N \tau_i - \frac{1}{M}\sum\limits_{i=1}^M \tau^{\text{sim}}_i \right| \geq \kappa \right]\\ 
&=    &&\mathbb{P}\left[\left|\frac{1}{N}\sum\limits_{i=1}^N \tau_i - \mathbb{E}[\tau] + \mathbb{E}[\tau] - \frac{1}{M}\sum\limits_{i=1}^M \tau^{\text{sim}}_i \right| \geq \kappa \right]\\ 
&\leq &&\mathbb{P}\left[ \left| \frac{1}{N}\sum\limits_{i=1}^N \tau_i - \mathbb{E}[\tau] \right| \geq \frac{\kappa}{2} \right] + \\ 
&{}  &&\mathbb{P}\left[ \left| \frac{1}{M}\sum\limits_{i=1}^M \tau^{\text{sim}}_i - \mathbb{E}[\tau] \right| \geq \frac{\kappa}{2} \right]\\
&\leq &&\frac{\eta}{2} + \frac{\eta}{2},
\end{align*}
where the last step follows with $M>N$ and Hoeffding's inequality for $\frac{\kappa}{2}$ and $\frac{\eta}{2}$.
\end{proof}
We avoid solving the PDE \eqref{PDE}, but the obtained trigger \eqref{approxtrigger} is less efficient than \eqref{trigger}, \ie the required sample size $N$ for equal accuracy $\eta$ is higher.

\begin{rem}[practical implementation]
In the experiments reported in the next sections, we use trigger \eqref{approxtrigger}, which is easier to compute than \eqref{trigger}.  In order to improve the trigger's robustness, we trigger learning experiments only when \eqref{approxtrigger} holds for a certain duration rather than instantaneous.
This way, the trigger is less prone to unmodeled short-term effects, which also decreases the probability of false positives.
%For the experiments we will use a moving average of size $N$ and additionally enforce the triggering rule to hold for a certain duration. The moving average is necessary to detect changes in the system, while the time requirement addresses the multiple comparisons problem and adds robustness to the trigger. 
\end{rem}

%% file: texs/6simulation.tex
\section{Simulation}
\label{sec:simulation}
This section illustrates the proposed event-triggered learning scheme with a numerical example.
%In this section, we demonstrate the derived theoretical results for a numerical example. 
%
For this, we consider two agents as in \fig \ref{multiagent} and a one-dimensional process
\begin{equation}
	x(k + 1) = 0.9 x(k) + 0.01 r(k) + \epsilon(k)
\end{equation}
with sample time $T_\text{s} = 0.01\, \text{s}$, reference $r(k) = \cos(0.2k T_\text{s})$, and $\epsilon(k) \sim \mathcal{N}(0, 2.5\cdot10^{-5})$.
To simulate imperfect model information, we assume that only a slightly perturbed model is available for making predictions,
\begin{equation}
	\hat{x}(k + 1) = 0.85 \hat{x}(k) + 0.005 r(k).
	\label{eq:simEx_model}
\end{equation}
%\begin{equation}
%	\hat{x}(k + 1) = \begin{cases} 0.85 \hat{x}(k) + 0.005 r(k) &\mbox{\text{if} } \gamma_{\text{state}}=0\\
%															 x(k + 1)	&\mbox{\text{if }}  \gamma_{\text{state}}=1 \end{cases}.
%\end{equation}

We implement all main components of the proposed event-triggered learning scheme as shown in \fig \ref{multiagent}; that is, the ETSE components as described in \sect \ref{sec:problem} with $\delta=0.02$,
%\ref{sec:ETSE_pred} and \ref{sec:ETSE_trig}, 
model learning as in \sect \ref{sec:modelLearning}, and the approximated learning trigger \eqref{approxtrigger}.
As trigger parameters, we chose a confidence level $\eta = 0.05$, as well as $N = 2000$, $M=10000$, and $\tau_{\mathrm{max}} = 3.5\, \text{s}$. Theorem \ref{thm:approxTrigger} then yields $\kappa_{\mathrm{approx}} \approx 0.23 $.
In order to obtain a robust learning trigger, $N$ in the empirical inter-communication time average, $\frac{1}{N}\sum_{i=1}^{N} \tau_i$, is reset after a learning experiment, and \eqref{approxtrigger} is only checked after having observed at least $N=2000$ stopping times $\tau_i$.

% ST: decided after discussion to leave out this aspect of what makes an "effective" trigger; interesting question for future work though!
%
%We demand a confidence level of $\eta = 0.01$; that is, we want the trigger to be effective with 99\%. 
%\commentseb{@Fr: This interpretation is correct, no?} 
%\commentfr{What do you mean by effective? If we choose $\eta = 0.01$ the prob for false positives is small but leads to $\kappa$ rather large for fixed $N$. Hence, there could be a mismatch but we do not trigger, since we are not confident due to tiny $\eta$. With $N \to \infty$ $\kappa$ will decrease and we detect change.} 

%We use trigger \eqref{approxtrigger} to decide if model learning is necessary. We set $\eta = 0.01$, 
%Observing values outside this interval is very unlikely and we decide to learn a new model if it happens.

%\commentseb{The next two sentences are left from my editing.  I don't think they are necessary, so I would suggest to delete them. (You have already explained before, how the trigger is supposed to work, so the experiment should only be about results, and possibly pointing this out again when looking at results, but not beforehand.  @Fr: if you agree, please remove.}

%During the experiment, we compare the observed mean intercommunication time with the approximation of the expected value \eqref{mean}. 

%The theoretical analysis in the previous section quantifies the maximally expected difference between the two values and yields a confidence interval. 

   \begin{figure}[tb]
      \centering
  	   \vspace{-3mm} %ST: a hack, usually, you should properly bound the image (ie. have no extra white space above and below the figure
			\includegraphics[scale=0.25]{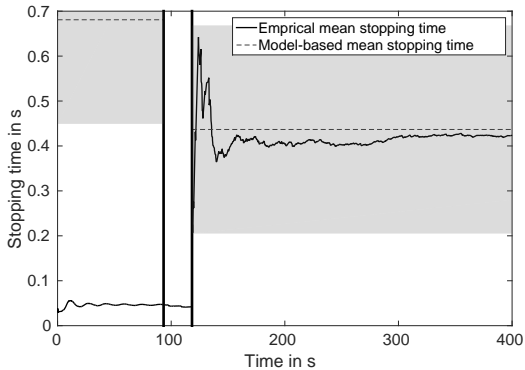}
			\vspace{-4mm}
      \caption{
      Communication and learning behavior for the simulation example. 
      The solid line shows the empirical inter-communication time 
      %(synonymous with stopping time) 
      computed as moving average $\frac{1}{N}\sum_{i=1}^{N} \tau_i$. The dashed line represents the model-based stopping time \eqref{mean} with the triggering interval $\pm \kappa_{\mathrm{approx}}$ in gray. The two vertical lines represent the start and the end of a learning experiment.  After learning, the empirical and model-based stopping times match well. 
  % Prev version:    
%      Stopping time analysis for a numerical example. The solid line shows the mean intercommunication time for a moving average of $N=2000$ samples. Based on the model we compute the expected mean stopping time (dotted line) and a confidence interval for $\eta = 0.01$ (gray area). The two vertical vertical lines represent the start and end of a learning experiment. After updating the model, the moving average is resetted. A new learning experiment can at the earliest be triggered after $N$ stopping times are observed.
%Whenever the empirical mean stopping time leaves the gray area for long enough, the trigger for model learning gets activated ($\gamma_{\text{learn}}=1$).
}
      \label{sinus}
			\vspace{-4mm}
   \end{figure}

Figure \ref{sinus} illustrates the behavior of the event-triggered estimation and learning system for this numerical example.
The actual communication is captured by the empirical inter-communication time, which is the inverse of the communication rate and computed as the moving average $\frac{1}{N}\sum_{i=1}^{N} \tau_i$, where $N$ is equal to $2000$ or the number of stopping times observed since the start or last learning experiment.
Since the initial model \eqref{eq:simEx_model} is inaccurate, we observe a significant deviation between empirical inter-communication time and model-based one \eqref{mean}.  
%The trigger \eqref{approxtrigger} is only checked after having observed $N=2000$ stopping times $\tau_i$ (and hence having a reliable estimate of the empirical stopping time), which happens at around $90 \, \text{s}$.  
At around $90 \, \text{s}$, 2000 inter-communication times $\tau_i$ have been observed, and the trigger \eqref{approxtrigger} is checked.
Because the triggering condition (right-hand side of \eqref{approxtrigger}) is clearly true, a learning experiment is performed (until about $120 \, \text{s}$).
Directly after learning, the moving average of the empirical inter-communication time is reset, hence its fluctuations.
%we observe fluctuations in the mean intercommunication time, which are caused through a reset in the moving average. 
%After we use the improved model for the prediction block we observe a rise in the intercommunication time. 
With the updated model, empirical and model-based stopping time coincide well; hence, no further learning experiment is triggered thereafter.%\footnote{\commentfr{Don't need this anymore, right?}Since model learning can be triggered again at the earliest after having observed $N=2000$ stopping times, there is no renewed learning around $125 \, \text{s}$ despite the empirical stopping time crossing the threshold.}  What is more, learning has reduced communication by a factor of $8$.

% prev version:
%It is in particular greater then $\kappa_{\mathrm{approx}}$ (empirical mean stopping time is outside the gray area) and hence, learning is triggered $(\gamma_{\mathrm{learn}} = 1)$. Directly after learning a new model, we observe fluctuations in the mean intercommunication time, which are caused through a reset in the moving average. Since model learning can at the earliest be triggered after observing $N$ intercommunication times, there is no renewed learning around $120$s. After we use the improved model for the prediction block we observe a rise in the intercommunication time. Moreover, observed and expected mean intercommunication times coincide now. 

%% file: texs/7experiment.tex
\section{Physical Experiment}
\label{sec:experiment}
We demonstrate the proposed event-triggered learning approach in experiments on a cart-pole system (see \fig \ref{cartpolesys}), which is a common benchmark in control \cite{Bou13}.  We consider balancing about the upright equilibrium, whose dynamics are approximately linear \eqref{LTI} with position $s$, velocity $\dot{s}$, angle $\theta$, angular velocity $\dot{\theta}$ as its states, and motor voltage as input $u$.  The dynamics are stabilized with a standard controller \eqref{eq:stateFeedback}, and we consider remote state estimation of the stabilized process \eqref{CL} as per the architecture in \fig \ref{multiagent}.

%\end{center}
%\centering
\begin{figure}[tb]
\vspace{1mm}
\centering
\includegraphics[width=2.8cm]{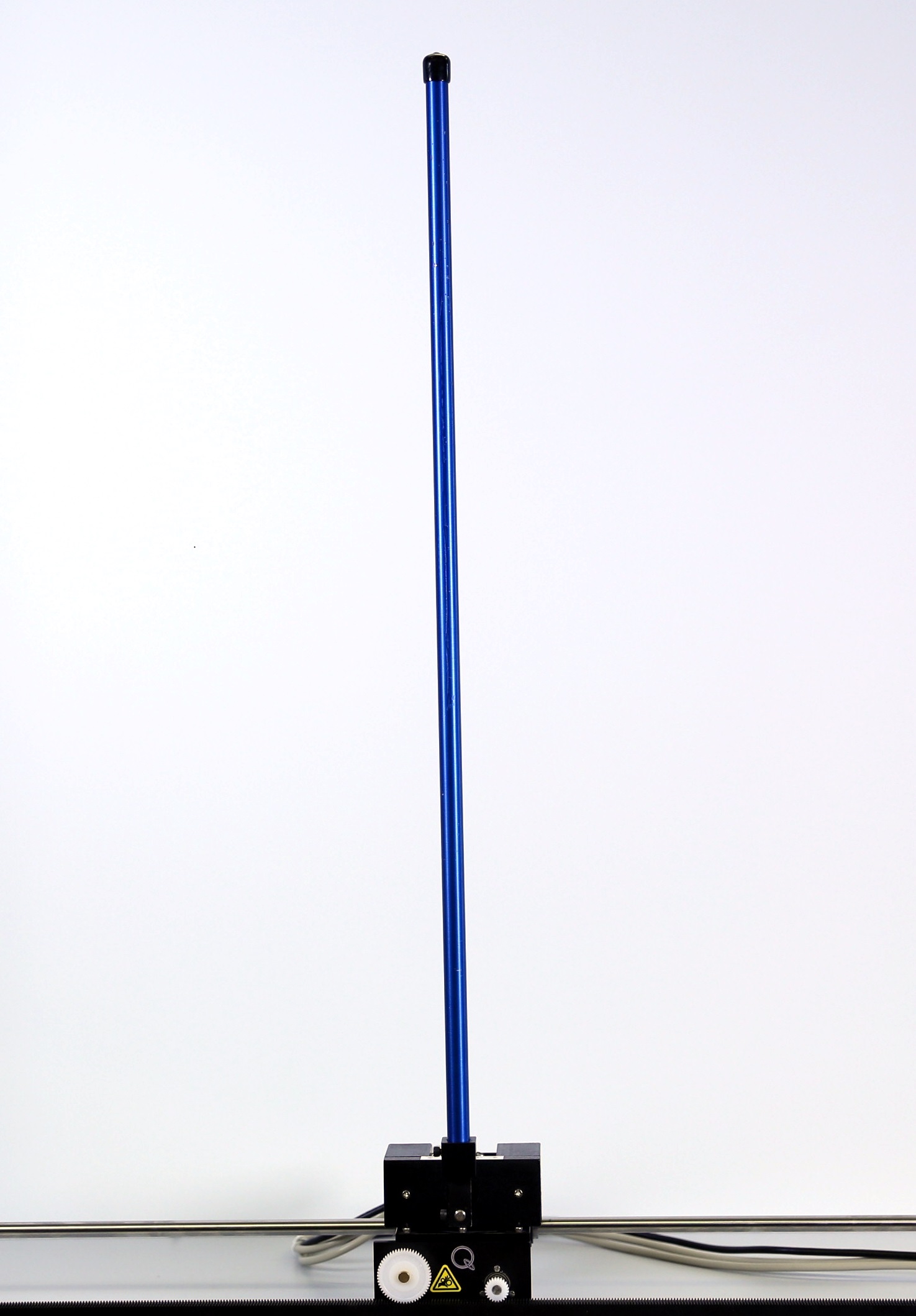} % was: scale=0.075
\hspace{2mm}
\includegraphics[width=2.8cm]{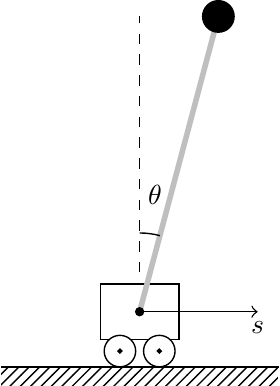}    % was: scale=1.3
%	\begin{subfigure}{0.5\columnwidth}
%      \centering
%			\includegraphics[width=3cm]{pictures/pendulum_lowquality} % was: scale=0.075
%	\end{subfigure}
%	\begin{subfigure}{}
%      \centering
%			\includegraphics[width=3cm]{pictures/Pole-Cart}    % was: scale=1.3
%	\end{subfigure}
	\caption{The event-triggered learning approach is demonstrated in experiments of the depicted cart-pole system.%Image and schematic drawing of the deployed cart-pole system.
	}
  \label{cartpolesys}
	\vspace{-4mm}
\end{figure}

% Prev version:
%We present the benefits of event-triggered learning for the remote state estimation problem on a cart-pole system. The system is illustrated in \fig \ref{cartpolesys} and can be modeled with the following states: position $(s(t))$, velocity $(\dot{s}(t))$, angle $(\theta(t))$, and angular velocity $(\dot{\theta}(t))$. Pole balancing is a challenging nonlinear problem, which serves as a benchmark for many control tasks (see e.\,g. \cite{Bou13}). The dynamics around $(\theta = 0)$ are approximately linear and the system can be stabilized with a static linear feedback controller. 

For making predictions in ETSE, we initially use a linear model supplied by the manufacturer \cite{Quanser}, and we take $\delta = 0.075$ as triggering threshold.  As the learning trigger, we use \eqref{approxtrigger}, where the empirical mean stopping time is computed over a fixed duration $T = 60 \, \text{s}$, which is an alternative to including a fixed number of samples in the average.  While the physical system does not exactly match the assumption in \sect \ref{sec:OUprocess} (\eg no strictly linear Gaussian dynamics), the theoretical analysis of \sect \ref{sec:eventTriggerLearning} is still helpful in guiding our choice of a triggering threshold $\kappa$. We choose  $\kappa = 2.5$ in the trigger \eqref{approxtrigger}, which is slightly lower then what we would obtain through Theorem \ref{thm:approxTrigger}. Since we additionally enforce the triggering condition to be true for $120 \, \text{s}$, we reduce the risk of false positives. 
% Prev version:
%We start with an initial linear model supplied by the manufacturer \cite{Quanser} and use this model to obtain state predictions. Simultaneously, we track the error \eqref{predictionerror} between true and predicted states. Whenever the error crosses the threshold $\delta = 0.075$, we communicate the true state and reset the prediction to this value. After balancing the pole for $240$s, we change the system by putting weights on top of the pole (red line in \fig \ref{experiment}). We use trigger \eqref{approxtrigger} to decide if model learning is necessary. The parameter $\kappa = 2.5$ was not chosen according to Theorem 2, but hand tuned in the interval $[\kappa_{\mathrm{exact}}, \kappa_{\mathrm{approx}}]$, which corresponds approximately to $[1.4\text{s}, 3\text{s}]$ for $\eta = 0.01$ and $\tau_\mathrm{max} = 11s$. After changing the system dynamics and learning a new model we use an even smaller $\kappa$, due to decreased stopping times.

Figure \ref{experiment} shows the results of the experiment. 
While the model used for predictions is not perfect, there is still a good match between the empirical and model-based stopping time: the empirical stopping time mostly remains within the gray area and thus no learning is triggered initially.  At $240$s, we add weights on top of the pole (\cf \fig \ref{experiment}), thus changing the system dynamics.  
The empirical inter-communication time drops (\ie indicating more communication) and clearly captures the change in dynamics. 
After the triggering condition being true (the black graph being outside the gray area) for $120 \, \text{s}$, a learning experiment is triggered.
%(at $370 \, \text{s}$).
%
% Omitted (prev version):
%The empirical mean stopping time is mostly inside the gray area and does not leave it for too long. After putting the weights on top of the pole, we observe a significant drop in the intercommunication time. After observing the empirical mean stopping time being outside the gray area for $120$s, we trigger a learning experiment.  
%This way we want assure that the deviations are caused by structural reasons and not through other unmodeled effects, which can occur in a physical system. For example around $140$s, the system was balancing almost perfectly which explains the unexpectedly high intercommunication time.   
A new model is then learned according to \sect \ref{sec:modelLearning} (an open-loop model is identified instead of a closed-loop one, which makes no difference here).
After updating the model, the empirical and model-based stopping times coincide very well. The event-triggered learning thus successfully detected and compensated for the changed dynamics, and reduced communication thereafter.

%\begin{center}
   \begin{figure}[tb]
   \centering
   \vspace{-3mm} %ST: a hack, usually, you should properly bound the image (ie. have no extra white space above and below the figure
			\includegraphics[scale=0.25]{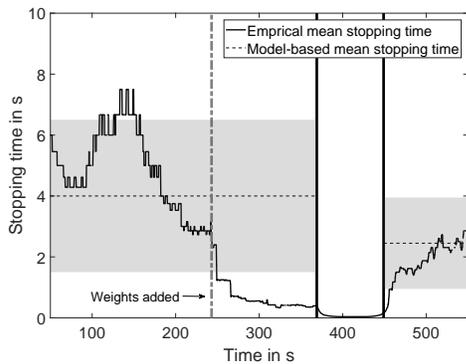}
			\vspace{-4mm}
      \caption{Communication and learning results for the cart-pole experiments. 
      Color coding for the graphs is the same as in \fig \ref{sinus}.
      The vertical dashed line indicates a change of the system dynamics (increasing the pole mass).}
  %    
      % Prev version:
      %Experimental results, showing the mean intercommunication time for the event-triggered state estimation problem on a cart pole system. The solid line shows the mean intercommunication time for a moving average of 60s. Based on the model we compute the expected mean stopping time (dotted horizontal line) and a confidence interval. We change the dynamics by adding weights on top of the pole (vertical dotted line). The two vertical black lines represent the start and end of a learning experiment.}
      \label{experiment}
			\vspace{-5mm}
   \end{figure}

%\addtolength{\textheight}{-3cm}   % This command serves to balance the column lengths

%\commentfr{Taking the moving average over the last $T$ seconds instead of the last $N$ stopping times is also a practical way to implement the trigger.}

%\commentseb{@Fr: I omitted ``After learning, an even smaller $\kappa$, due to decreased stopping times.'' -- not sure if you still need it.}

%% file: texs/8rest.tex
\section{Conclusion}
%\section{Conclusion and Future Works}
Event-triggered learning is proposed in this paper as a novel concept to trigger model learning when needed.  The concept is applied to event-triggered state estimation (ETSE) and shown to lead to reduced communication in simulation and physical experiments.  For this setting, we obtained (provably) effective learning triggers by means of statistical stopping time analysis. 

This paper is the first to develop the concept of event-triggered learning. Extending the method to nonlinear dynamic systems
%, \eg through Gaussian process model learning \cite{RaWi06,doerr_corl_2017}, 
is an obvious next step we plan for future work.  
While event-triggered learning has been motivated as an extension to ETSE, the concept generally addresses the fundamental question of \emph{when to learn}, and thus 
%addresses the exploration/exploitation fundamental to learning control and thus 
potentially has much wider relevance.  
%This is an interesting direction to explore in further research.  Other topics we plan to address in future work

% prev version:
%Experimental and numerical results demonstrate that the developed event-triggered learning algorithm is
%an effective tool for reducing the average communication
%rate in an event-triggered estimation scheme even further. Moreover, it is a
%straight-forward tool for detecting change in the dynamics online with promising applications in networked systems.
%The presented work leaves room for elaborating the framework for nonlinear dynamics and including the filtering problem.  

%\todoseb{Mention other model learning; cite Andreas' paper?  Exploration-exploitation...}

%%%%%%%%%%%%%%%%%%%%%%%%%%%%%%%%%%%%%%%%%%%%%%%%%%%%%%%%%%%%%%%%%%%%%%%%%%%%%%%%

%%%%%%%%%%%%%%%%%%%%%%%%%%%%%%%%%%%%%%%%%%%%%%%%%%%%%%%%%%%%%%%%%%%%%%%%%%%%%%%%
\section*{Acknowledgment}
The authors thank Felix Grimminger for his support with the cart-pole system, and Alonso Marco for his help with speeding up the learning experiments.

%% file: root.bbl
\begin{thebibliography}{10}
\providecommand{\url}[1]{#1}
\csname url@samestyle\endcsname
\providecommand{\newblock}{\relax}
\providecommand{\bibinfo}[2]{#2}
\providecommand{\BIBentrySTDinterwordspacing}{\spaceskip=0pt\relax}
\providecommand{\BIBentryALTinterwordstretchfactor}{4}
\providecommand{\BIBentryALTinterwordspacing}{\spaceskip=\fontdimen2\font plus
\BIBentryALTinterwordstretchfactor\fontdimen3\font minus
  \fontdimen4\font\relax}
\providecommand{\BIBforeignlanguage}[2]{{%
\expandafter\ifx\csname l@#1\endcsname\relax
\typeout{** WARNING: IEEEtran.bst: No hyphenation pattern has been}%
\typeout{** loaded for the language `#1'. Using the pattern for}%
\typeout{** the default language instead.}%
\else
\language=\csname l@#1\endcsname
\fi
#2}}
\providecommand{\BIBdecl}{\relax}
\BIBdecl

\bibitem{HespanhaNaghshtabriziXuJan07}
J.~P. Hespanha, P.~Naghshtabrizi, and Y.~Xu, ``A survey of recent results in
  networked control systems,'' \emph{Proc. IEEE \textnormal{Special Issue on
  Technology of Networked Control Systems}}, vol.~95, no.~1, pp. 138--162, Jan.
  2007.

\bibitem{cao2013overview}
Y.~Cao, W.~Yu, W.~Ren, and G.~Chen, ``An overview of recent progress in the
  study of distributed multi-agent coordination,'' \emph{IEEE Trans. on
  Industrial informatics}, vol.~9, no.~1, pp. 427--438, 2013.

\bibitem{TrDAn12b}
S.~Trimpe and R.~D'Andrea, ``The {B}alancing {C}ube: A dynamic sculpture as
  test bed for distributed estimation and control,'' \emph{IEEE Control Systems
  Magazine}, vol.~32, no.~6, pp. 48--75, Dec. 2012.

\bibitem{YoTiSo02}
J.~K. Yook, D.~M. Tilbury, and N.~R. Soparkar, ``Trading computation for
  bandwidth: reducing communication in distributed control systems using state
  estimators,'' \emph{IEEE Transactions on Control Systems Technology},
  vol.~10, no.~4, pp. 503--518, Jul. 2002.

\bibitem{TrDAn11}
S.~Trimpe and R.~D'Andrea, ``An experimental demonstration of a distributed and
  event-based state estimation algorithm,'' in \emph{18th IFAC World Congress},
  2011, pp. 8811--8818.

\bibitem{sijs2012event}
J.~Sijs and M.~Lazar, ``Event based state estimation with time synchronous
  updates,'' \emph{IEEE Transactions on Automatic Control}, vol.~57, no.~10,
  pp. 2650--2655, 2012.

\bibitem{wu2013event}
J.~Wu, Q.-S. Jia, K.~H. Johansson, and L.~Shi, ``Event-based sensor data
  scheduling: Trade-off between communication rate and estimation quality,''
  \emph{IEEE Transactions on Automatic Control}, vol.~58, no.~4, pp.
  1041--1046, 2013.

\bibitem{TAC14b_web}
S.~Trimpe and R.~D'Andrea, ``Event-based state estimation with variance-based
  triggering,'' \emph{IEEE Transactions on Automatic Control}, vol.~59, no.~12,
  pp. 3266--3281, 2014.

\bibitem{shi2015event}
D.~Shi, L.~Shi, and T.~Chen, \emph{Event-Based State Estimation: A Stochastic
  Perspective}.\hskip 1em plus 0.5em minus 0.4em\relax Springer, 2015.

\bibitem{trimpeEBCCSP15}
S.~Trimpe and M.~C. Campi, ``On the choice of the event trigger in event-based
  estimation,'' in \emph{Int. Conf. on Event-based Control, Communication, and
  Signal Processing}, 2015, pp. 1--8.

\bibitem{SiNoLaHa16}
J.~Sijs, B.~Noack, M.~Lazar, and U.~D. Hanebeck, ``Time-periodic state
  estimation with event-based measurement updates,'' in \emph{Event-Based
  Control and Signal Processing}.\hskip 1em plus 0.5em minus 0.4em\relax CRC
  Press, 2016.

\bibitem{tr17}
S.~Trimpe, ``Event-based state estimation: An emulation-based approach,''
  \emph{IET Control Theory \& Applications}, vol.~11, no.~11, pp. 1684--1693,
  Jul. 2017.

\bibitem{miskowicz2006send}
M.~Miskowicz, ``Send-on-delta concept: an event-based data reporting
  strategy,'' \emph{Sensors}, vol.~6, no.~1, pp. 49--63, 2006.

\bibitem{Su07}
Y.~S. Suh, ``Send-on-delta sensor data transmission with a linear predictor,''
  \emph{Sensors}, vol.~7, no.~4, pp. 537--547, 2007.

\bibitem{SiKeNo14}
J.~Sijs, L.~Kester, and B.~Noack, ``A study on event triggering criteria for
  estimation,'' in \emph{17th International Conference on Information Fusion},
  July 2014, pp. 1--8.

\bibitem{MuTr18}
M.~Muehlebach and S.~Trimpe, ``Distributed event-based state estimation for
  networked systems: An lmi approach,'' \emph{IEEE Transactions on Automatic
  Control}, vol.~63, no.~1, pp. 269--276, Jan. 2018.

\bibitem{TrBu15}
S.~Trimpe and J.~Buchli, ``Event-based estimation and control for remote robot
  operation with reduced communication,'' in \emph{IEEE Int. Conf. on Robotics
  and Automation}, 2015, pp. 5018--5025.

\bibitem{MaEsGaSa15}
M.~Mart{\'\i}nez-Rey, F.~Espinosa, A.~Gardel, and C.~Santos, ``On-board
  event-based state estimation for trajectory approaching and tracking of a
  vehicle,'' \emph{Sensors}, vol.~15, no.~6, pp. 14\,569--14\,590, 2015.

\bibitem{XuHe04}
Y.~Xu and J.~P. Hespanha, ``Communication logics for networked control
  systems,'' in \emph{American Control Conference}, 2004, pp. 572--577.

\bibitem{lai1995sequential}
T.~L. Lai, ``Sequential changepoint detection in quality control and dynamical
  systems,'' \emph{Journal of the Royal Statistical Society. Series B
  (Methodological)}, pp. 613--658, 1995.

\bibitem{wang2008new}
Y.~Wang, S.~X. Ding, H.~Ye, and G.~Wang, ``A new fault detection scheme for
  networked control systems subject to uncertain time-varying delay,''
  \emph{IEEE Transactions on Signal Processing}, vol.~56, no.~10, pp.
  5258--5268, 2008.

\bibitem{haghighi2013dynamic}
M.~Haghighi and C.~J. Musselle, ``Dynamic collaborative change point detection
  in wireless sensor networks,'' in \emph{Int. Conf. on Cyber-Enabled
  Distributed Computing and Knowledge Discovery}, 2013, pp. 332--339.

\bibitem{zhang2016event}
T.~Zhang and J.~Li, ``Event-triggered iterative learning control for
  multi-agent systems with quantization,'' \emph{Asian Journal of Control},
  2016.

\bibitem{ljung1999system}
L.~Ljung, \emph{System identification: Theory for the User}.\hskip 1em plus
  0.5em minus 0.4em\relax Prentice Hall, New Jersey, 1999.

\bibitem{oksendal2003stochastic}
B.~{\O}ksendal, ``Stochastic differential equations,'' in \emph{Stochastic
  differential equations}.\hskip 1em plus 0.5em minus 0.4em\relax Springer,
  2003, pp. 65--84.

\bibitem{Sc09}
Z.~Schuss, \emph{Theory and Applications of Stochastic Processes: An Analytical
  Approach}, ser. Applied Mathematical Sciences.\hskip 1em plus 0.5em minus
  0.4em\relax Springer New York, 2009.

\bibitem{aastrom2012introduction}
K.~J. {\AA}str{\"o}m, \emph{Introduction to stochastic control theory}.\hskip
  1em plus 0.5em minus 0.4em\relax Dover Publications, 2006.

\bibitem{von2008statistical}
U.~von Luxburg and B.~Sch{\"o}lkopf, \emph{Statistical Learning Theory: Models,
  Concepts, and Results}.\hskip 1em plus 0.5em minus 0.4em\relax Amsterdam,
  Netherlands: Elsevier North Holland, May 2011, vol.~10, pp. 651--706.

\bibitem{Bou13}
O.~Boubaker, ``The inverted pendulum benchmark in nonlinear control theory: A
  survey,'' \emph{International Journal of Advanced Robotic Systems}, vol.~10,
  no.~5, p. 233, 2013.

\bibitem{Quanser}
{Quanser Inc.}, ``{IP}02 - self-erecting single inverted pendulum ({SESIP}) -
  linear experiment \#6: {PV} and {LQR} control - instructor manual.''

\end{thebibliography}
